\documentclass[aip, jmp, 12pt]{revtex4-1}
\usepackage[utf8]{inputenc}
\usepackage{amsmath}
\usepackage{amsfonts}
\usepackage{amsthm}
\usepackage{mathtools}
\usepackage{hyperref}
\usepackage{cleveref}
\usepackage{stmaryrd}
\usepackage{thmtools}
\usepackage{thm-restate}
\usepackage{bbm}
\usepackage{graphicx}
\usepackage{epstopdf}

\crefname{secinapp}{appendix}{appendices}
\Crefname{secinapp}{Appendix}{Appendices}

\newtheorem{definition}{Definition}
\newtheorem{lemma}{Lemma}
\newtheorem{theorem}{Theorem}
\newtheorem*{remark}{Remark}
\newtheorem{corollary}{Corollary}

\newcommand{\size}[1]{\left\lvert #1 \right\rvert}
\newcommand{\graph}[2]{\mathbb{#1}_{#2}}
\newcommand{\intset}[2]{\left \llbracket #1 , #2 \right \rrbracket}
\newcommand{\expected}[2]{\left\langle #1 \right\rangle_{#2}}

\begin{document}

\title{The Fundamental Gap for a Class of Schr\"{o}dinger Operators on Path and Hypercube Graphs}
\author{Michael Jarret}
\affiliation{Department of Physics, University of Maryland, College Park, MD 20742-4111}
\email{mjarret@umd.edu}
\author{Stephen P. Jordan}
\affiliation{Applied and Computational Mathematics Division, National Institute of Standards and Technology,  Gaithersburg,  MD 20899}
\email{stephen.jordan@nist.gov}

\begin{abstract}
  We consider the difference between the two lowest eigenvalues (the fundamental gap) of a Schr\"{o}dinger operator acting on a class of graphs. In particular, we derive tight bounds for the gap of Schr\"{o}dinger operators with convex potentials acting on the path graph. Additionally, for the hypercube graph, we derive a tight bound for the gap of Schr\"{o}dinger operators with convex potentials dependent only upon vertex Hamming weight. Our proof makes use of tools from the literature of the fundamental gap theorem as proved in the continuum combined with techniques unique to the discrete case. We prove the tight bound for the hypercube graph as a corollary to our path graph results. 
\end{abstract}

\maketitle

\section{Introduction}
  The Fundamental Gap Conjecture proposed a tight lower bound of $3\pi^2/D^2$ to the difference between the two lowest eigenvalues (the gap) of a Schr\"{o}dinger operator $-\nabla^2 + V(x)$ with convex potential $V$ on a compact convex domain $\Omega \subset \mathbb{R}^n$ of diameter $D$ and subject to Dirichlet boundary conditions. Recently, Andrews and Clutterbuck proved the conjecture for all ``semiconvex'' potentials (which include convex potentials as a special case) in arbitrary dimensions \cite{andrews2011proof}. Although the community's focus has largely centered on the continuum\cite{andrews2011proof, Lavine1994, ashbaugh1989optimal, Yu1986}, as early as 1990 Ashbaugh and Benguria saw the potential for extending their results to discrete Laplacians. In their work, they proved a lower bound to the gap for a particular class of discrete Laplacians with symmetric-decreasing potentials \cite{ashbaugh1990some}. Indeed, recent interest in adiabatic quantum computing justifies their vision and motivates our interest in lifting continuum results to graph Laplacians\cite{Farhi_science, FGG02}. While our interest is driven by quantum computation, the discrete eigenvalue gap is also of interest to condensed matter physicists. Abstractly, this result is a useful addition to spectral theory.

  Previously, in the setting of quantum computation, gap bounds were derived on an as-needed basis. For instance, in an analysis of the power of adiabatic algorithms, van Dam et al. bounded eigenvalue gaps in the minimum Hamming weight problem by considering an explicit gap and then bounding the maximum error on this gap from perturbations\cite{DMV01}. In another instance, Reichardt considers the eigenvalue gap for an Ising system by using properties of the operator's principal submatrices\cite{R04}. (At least in the case of the path graph, Reichardt's Sturm sequences are similar in form to our eigenvector recurrence of \cref{eqn:recurrence}. For an explicit examination of the link between principal submatrices and the eigenvector recurrences, see Gantmakher and Kre\u{i}n\cite{gantmakher2002oscillation}.) Unlike the constructions above, we look to develop tools of increasingly general applicability. Thus, we begin with systems where gaps are demonstrably ``large'' and search for extensions of these systems to problems of algorithmic and physical interest.

  In this work, we consider specifically Schr\"{o}dinger operators corresponding to graph Laplacians with suitably defined convex potential terms. Here, the potential is restricted to the vertices and can be seen either as a site-dependent physical potential (as in the physics literature) or as a weighted graph with loops (as in the mathematical and computer science literature). Thus, for a graph $\mathbb{G}=(V,E)$ with graph Laplacian $\mathbf{L(\mathbb{G})}$ and subjected to a potential $W(\cdot)$ we consider Schr\"{o}dinger operators of the form
  \begin{equation}
    \mathbf{H}_{W}(\mathbb{G}) = \mathbf{L(\mathbb{G})} + \mathbf{W}
  \end{equation}
  where
  \begin{equation}
    \big[\mathbf{W}(V)\big]_{ij} = W(V_i) \delta_{ij}.
  \end{equation}
  Although our problem is analogous to the Fundamental Gap Conjecture as proven in the continuum, lifting existing results to the discrete realm and maintaining tight bounds is non-trivial. Perhaps the most obvious challenge we face is the loss of well-defined boundary conditions and, for this reason, we restrict our initial study to the path and hypercube graphs. In the first case, our restriction gives boundary conditions similar to Neumann boundary conditions in the continuum and thus our result bears some resemblance to the continuum one of Payne and Weinberger \cite{Payne1960} and indeed converges upon this result asymptotically. (For the physicist, our path graph Hamiltonian can be viewed as a 1-dimensional chain with a nearest-neighbor interaction term and a convex, site-dependent potential term. See \Cref{fig:path}. Up to an identity term, the Laplacian of the hypercube graph of $2^N$ vertices,  $\mathbb{H}_{2^N}$, is equivalent to a sum of the Pauli $\sigma_x$ operators acting on each of $N$ qubits. In particular, transverse Ising models such as those studied in \cite{Farhi_science} can be cast as potentials on the hypercube. Here, like Reichardt\cite{R04} and van Dam et al. \cite{DMV01}, but unlike Farhi et al.\cite{Farhi_science}, we focus on the case that the potential depends only on the Hamming distance from a minimum. For the hypercube graph see \Cref{fig:path}.)
  
  \begin{figure}[htp]
    \includegraphics[width=0.78\textwidth]{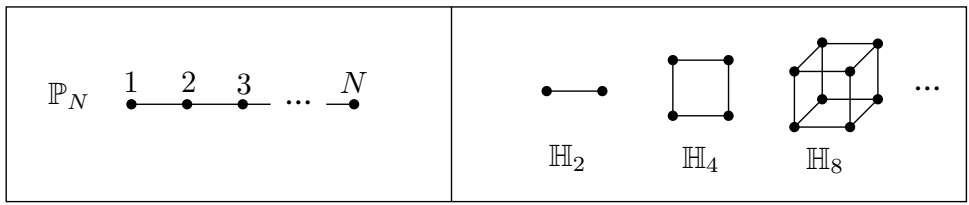}
    \caption{\label{fig:path} The path graph $\graph{P}{N}$ of length $N$ and the first three hypercube graphs $\graph{H}{2}$,$\graph{H}{4}$, and $\graph{H}{8}$.}
  \end{figure}

  In particular, we show that for convex potentials on the path graph $\mathbb{P}_N$ of length $N$ the gap $\Gamma$ is bounded by the gap corresponding to the flat potential
  \begin{equation}
    \Gamma \geq 2\left(1-\cos\left(\frac{\pi}{N}\right)\right).
  \end{equation}
  On the hypercube graph $\graph{H}{2^N}$, for convex potentials dependent only upon vertex Hamming weight, we prove a similar flat-potential lower bound given by
    \begin{equation}
    \Gamma \geq 2.
  \end{equation}

\section{Preliminaries}\label{sec:prelim}
  \subsection{The graph Laplacian and its Eigenvalues}\label{sec:prelim_eigenvalues}
    Let $\mathbb{G}=(V,E)$ be an undirected graph with vertex set $V$ and edge set $E\subseteq V \times V$. Then we associate with $\mathbb{G}$ a degree matrix $\mathbf{D}(\mathbb{G})$ and an adjacency matrix $\mathbf{A}(\mathbb{G})$ where
    \begin{equation}
      \big[\mathbf{D}(\mathbb{G})\big]_{ij} = d_i \delta_{ij}
    \end{equation}
    with $d_i$ the degree of vertex $V_i \in V$ and
    \begin{equation}
      \big[\mathbf{A}(\mathbb{G})\big]_{ij} = \begin{cases}
						1 & \text{if $(V_i,V_j) \in E$} \\
						0 & \text{otherwise}.
					      \end{cases}
    \end{equation}

    One then defines the $\lvert V \rvert \times \lvert V \rvert$ graph Laplacian $\mathbf{L}(\mathbb{G})$ as the difference between the degree matrix and adjacency matrix. That is,
    \begin{equation}
      \mathbf{L}(\mathbb{G}) = \mathbf{D}(\mathbb{G}) - \mathbf{A}(\mathbb{G}).
    \end{equation}

    We now extend our attention to a more general class of Schr\"{o}dinger operators of the form
    \begin{equation}
      \mathbf{H}_W(\mathbb{G}) \stackrel{\text{def}}{=} \mathbf{L}(\mathbb{G}) + \mathbf{W}(V) \label{eqn:hw}
    \end{equation}
    where for some function $W:V\rightarrow\mathbb{R}$, $\mathbf{W}$ is the diagonal matrix defined by
    \begin{equation}
      \big[\mathbf{W}(V)\big]_{ij} \stackrel{\text{def}}{=} W(V_i) \delta_{ij}.
    \end{equation}
    We can think of the resulting matrix as either the graph Laplacian for a weighted graph with loops or as a Schr\"{o}dinger operator (Hamiltonian) with an external potential. The eigenvalue spectrum of $\mathbf{H}_W(\mathbb{G})$ is $\lambda_1 \leq \lambda_2 \leq \dots \leq \lambda_{\size{V}}$ with associated, normalized eigenvectors $\mathbf{u}(\lambda_1),\mathbf{u}(\lambda_2),\dots,\mathbf{u}(\lambda_{\size{V}})$. Suppose now that we consider the one parameter family $\mathbf{H}_W(\mathbb{G};\alpha)$ with
    \begin{equation}
      \mathbf{H}_W(\mathbb{G}) = \mathbf{H}_W(\mathbb{G};\alpha) \Big\rvert_{\alpha=0}.
    \end{equation}
    If $\lambda_k$ is an eigenvalue of $\mathbf{H}_W(\mathbb{G};\alpha)$ with no degeneracy, the Hellman-Feynman theorem governs the relationship between $\lambda_k$ and $\alpha$. That is,
    \begin{theorem}[Hellman-Feynman]\label{thm:Hellman-Feynman}
      Let $\mathbf{H(\alpha)}$ be a Hermitian operator (matrix) dependent upon a parameter $\alpha$ with non-degenerate eigenvalue $\lambda(\alpha)$ and associated eigenvector $\mathbf{u}(\lambda;\alpha)$. Then
    \begin{equation}
      \frac{d \lambda(\alpha)}{d \alpha}  = \sum_{i,j} u_i^*(\lambda;\alpha) \frac{d \big[\mathbf{H(\alpha)}\big]_{ij}}{d \alpha} u_j(\lambda;\alpha) \equiv \expected{\frac{d \big[\mathbf{H(\alpha)}\big]_{ij}}{d \alpha}}{\mathbf{u}(\lambda;\alpha)}
    \end{equation}
    where $u_i(\lambda;\alpha)$ is the $i^{th}$ component of $\mathbf{u}(\lambda;\alpha)$. \footnote{It should be noted that the theorem is typically stated for a Hermitian operator $\mathbf{H}_\alpha$ with eigenvalues $\lambda_0 < \lambda_1 < \dots < \lambda_N$. Care must be taken in the application of this theorem when considering degenerate eigenvalues \cite{Vatsya2004,Zhang2004} which can occur in broader classes of graph Laplacians. For instance, the ring graph with constant potential has degeneracies. Nonetheless, since the cases we consider in this paper are non-degenerate, we can use this theorem in its above-stated form.}
    \end{theorem}

    Our primary interest in this paper is the so-called Fundamental Gap,
    \begin{equation}
      \Gamma(\alpha) \stackrel{\text{def}}{=} \lambda_2(\alpha) - \lambda_1(\alpha)
    \end{equation}
    the difference between the two lowest eigenvalues of $\mathbf{H}_W(\mathbb{G};\alpha)$. Assuming that both $\lambda_1$ and $\lambda_2$ are non-degenerate eigenvalues, by \Cref{thm:Hellman-Feynman} we have that
    \begin{equation}\label{eqn:Hellman-Feynman}
      \frac{d \Gamma(\alpha)}{d \alpha} = \expected{\frac{d \mathbf{H}_W(\mathbb{G};\alpha)}{d \alpha}}{\mathbf{u}(\lambda_2)} - \; \expected{\frac{d \mathbf{H}_W(\mathbb{G};\alpha)}{d \alpha}}{\mathbf{u}(\lambda_1)}
    \end{equation}
    where if we consider $\mathbf{H}_W(\mathbb{G};\alpha) = \mathbf{H}_{\alpha W}(\mathbb{G})$,
    \begin{equation}
      \frac{d \Gamma(\alpha)}{d \alpha} = \expected{\mathbf{W}}{\mathbf{u}(\lambda_2)} - \expected{\mathbf{W}}{\mathbf{u}(\lambda_1)}.
    \end{equation}

  \subsection{Eigenvectors of $\mathbf{H}_W(\mathbb{G})$}
    In deriving bounds for $\Gamma$ we make extensive use of the recurrence relations satisfied by the eigenvectors of $\mathbf{H}_W(\mathbb{G})$. Expressing the eigenvalue equation
    \begin{equation}
      \mathbf{H}_W(\mathbb{G}) \mathbf{u}(\lambda) - \lambda \mathbf{u}(\lambda)=0
    \end{equation}
    componentwise, we obtain the following set of linear equations.
    \begin{equation}\label{eqn:recurrence_graph}
      (d_i + W_i - \lambda)u_i(\lambda) = \sum_{(V_i,V_j)\in E} u_j(\lambda) \;\; \text{for $V_i \in V$}
    \end{equation}
    where for simplicity we let $W_i = W(V_i)$.

    When $\mathbb{G}$ is the path graph, we always consider the labeling of $V$ such that $(V_i,V_j) \in E \implies j = i\pm1$. Then, \cref{eqn:recurrence_graph} reduces to
    \begin{equation}\label{eqn:recurrence}
      (2 + W_i - \lambda)u_i(\lambda) = u_{i-1}(\lambda) + u_{i+1}(\lambda) \;\; \text{for $V_i \in V$.}
    \end{equation}
    Here, to simplify the treatment, we introduce fictitious vertices $u_0(\lambda)$ and $u_{\size{V}+1}(\lambda)$. We correspondingly set $u_0(\lambda)=u_1(\lambda)$ and $u_{\size{V}+1}(\lambda)=u_{\size{V}}(\lambda)$ for the path graph.

    For our purposes, it is often convenient to express \cref{eqn:recurrence} in terms of difference equations. For this, we need the forward difference operator.
    \begin{definition}[Forward Difference Operator]\label{def:forward_difference}
      For a given sequence $\left(u_i\right)$, we define $\Delta$, the forward difference operator, by $\Delta u_i = u_{i+1} - u_i$. We further define $\Delta^2$, the second difference operator, by $\Delta^2 u_i = u_{i+1} - 2 u_i + u_{i-1}$.
    \end{definition}
       It is also useful to note that for any sequence $(u_i)$,
      \begin{equation}\label{eqn:forward_sum}
	\sum_{i=a}^b \Delta u_i = u_{b+1} - u_a.
      \end{equation}

    \begin{remark}
      The reader should note that our notation yields $\Delta \left(\Delta u_i \right) \neq \Delta^2 u_i$. This makes $\Delta^2$ a central difference operator, not a forward difference operator. This choice is convenient, since it allows us to easily keep track of indices as seen below in \cref{eqn:recurrence_difference}.
    \end{remark}

    Now, applying \Cref{def:forward_difference}, \cref{eqn:recurrence} becomes
    \begin{equation}\label{eqn:recurrence_difference}
      \Delta^2 u_i(\lambda) = (W_i - \lambda) u_i(\lambda)
    \end{equation}
    which, similar to the second derivative of a continuous function, is an expression of the convexity of $\mathbf{u}$ at $u_i$.

    We now define some other useful properties of sequences, which we will apply to both sequences and vectors without restatement.
    \begin{definition}[Generalized Zero]\label{def:generalized_zero}
      For a given sequence $\left(u_i\right)$ we call $u_m \in \left(u_i\right)$ a generalized zero if $u_m u_{m+1} < 0$ or $u_m = 0$.
    \end{definition}
    \begin{definition}
      For a given sequence $(u_i)$ we call the piecewise linear curve connecting Cartesian coordinates $(i,u_i)$ the $\mathbf{u}$-line.
    \end{definition}
    \begin{definition}\label{def:node}
      For a given sequence $(u_i)$ we call a point at which the $\mathbf{u}-line$ intersects zero a node and label it by its $x$-coordinate. From \Cref{def:generalized_zero} if $u_m \in (u_i)$ is a generalized zero, then the $\mathbf{u}$-line has a node at $x$ with $x\in[m,m+1)$.
    \end{definition}

    For two sequences $(u_i),(v_i)$ we will frequently need the discrete analogue of the Wronskian, the Casoratian sequence $(w_i)$. Suppose that $\mathbf{u}(\mu;\beta),\mathbf{u}(\lambda;\alpha)$ are two sequences (vectors) with $\mu > \lambda$, satisfying \cref{eqn:recurrence}, and parameterized by $\beta$ and $\alpha$ respectively. Then, we are interested in
    \begin{equation}
      w_i\big(\mathbf{u}(\mu;\beta),\mathbf{u}(\lambda;\alpha)\big) = u_{i+1}(\mu;\beta)u_i(\lambda;\alpha) - u_i(\mu;\beta)u_{i+1}(\lambda;\alpha)
    \end{equation}
    which, when applied to \cref{eqn:recurrence} yields
    \begin{equation}\label{eqn:Casoratian}
      \Delta w_{i-1}\big(\mathbf{u}(\mu;\beta),\mathbf{u}(\lambda;\alpha)\big) = \Theta_{W,i}(\mu-\lambda;\beta,\alpha) u_{i}(\mu;\beta)u_i(\lambda;\alpha)
    \end{equation}
    where
    \begin{equation}\label{eqn:theta}
      \Theta_{W,i}(\gamma;\beta,\alpha) \stackrel{\text{def}}{=} W_i(\beta) - W_i(\alpha) - \gamma.
    \end{equation}

\section{The Path Graph $\mathbb{P}_{N}$}\label{sec:path}

  For the path graph $\mathbb{P}_N$ depicted in \Cref{fig:path}, we are interested in the case of convex potentials, for which we offer the following definition:
    \begin{definition}\label{def:convex}
      Let $\intset{a}{b} = \{a,a+1,\dots,b-1,b\}$. Let $\mathbb{P}_N$ be the path graph with vertex set $V=\{V_i\}_{i\in \intset{1}{N}}$ and edge set $E=\big\{(V_i,V_{i+1})\big\}_{i\in\intset{1}{N-1}}$. Let $\mathcal{W}$ be the set of all convex functions $w:\mathbb{R}\rightarrow\mathbb{R}$. We call $W:V\rightarrow \mathbb{R}$ convex if there exists some $w \in \mathcal{W}$ such that $W(V_i) = w(i) \; \forall \; V_i \in V$.

      We similarly define the term ``linear'' and denote its set $\mathcal{L}$.
    \end{definition}

  We begin by using variational arguments to demonstrate that the gap corresponding to each $W\in\mathcal{W}$ is bounded from below by the gap corresponding to some $L\in\mathcal{L}$. This approach is modeled on that used by Lavine in the continuum.\cite{Lavine1994} Then, we use the geometry of the eigenvectors of $\mathbf{H}_L(\mathbb{P}_N)$ to demonstrate that the gap of each linear potential $L$ is bounded from below by the gap for a constant potential.

  \subsection{The gap for convex potentials is lower bounded by the gap for linear potentials.}
  The eigenvalues of $\mathbf{H}_W(\mathbb{P}_N)$ are real and ordered $\lambda_1 < \lambda_2 < \dots < \lambda_N$. Also, recall that we have introduced fictitious points $u_0(\lambda)$ and $u_{N+1}(\lambda)$ to satisfy the recurrence \cref{eqn:recurrence}. Then we have the following fact about the intersections of the $\mathbf{u}(\lambda_1)$-line and $\mathbf{u}(\lambda_2)$-line.
  \begin{lemma}\label{lem:intersection}
    Let $0\leq\lambda_1<\lambda_2$ be the two lowest eigenvalues of $\mathbf{H}_{W}(\mathbb{P}_N)$ for convex $W$, and let $\mathbf{u}(\lambda_1),\mathbf{u}(\lambda_2)$ be their corresponding eigenvectors. Then, $\exists m<n \in \intset{1}{N}$ such that $u_{i}^2(\lambda_2) - u_{i}^2(\lambda_1) \geq 0$ for all $i \in \intset{1}{m}\cup\intset{n+1}{N}$ and $u_{i}^2(\lambda_2) - u_{i}^2(\lambda_1) < 0$ for all $i \in \intset{m+1}{n}$.
  \end{lemma}
  \begin{proof}
    The intersections of $\mathbf{u}(\lambda_1)$ and $\mathbf{u}(\lambda_2)$ can be characterized by the behavior of the quantity
    \begin{align}
      \Delta\left(\frac{u_i(\lambda_2)}{u_i(\lambda_1)} \right) &= \frac{u_{i+1}(\lambda_2)u_i(\lambda_1)-u_i(\lambda_2)u_{i+1}(\lambda_1)}{u_{i+1}(\lambda_1)u_i(\lambda_1)} \\
								&\equiv \frac{w_i\big(\mathbf{u}(\lambda_2),\mathbf{u}(\lambda_1)\big)}{u_{i+1}(\lambda_1)u_i(\lambda_1)}. \label{eq:cas_temp}
    \end{align}

    For simplicity, let $w_i = w_i\big(\mathbf{u}(\lambda_2),\mathbf{u}(\lambda_1)\big)$. Then, in \cref{eqn:Casoratian} we can set $\alpha=\beta=0$, yielding
    \begin{equation}\label{eqn:Casoratian_diff}
      \Delta w_{i-1} = -\Gamma u_i(\lambda_2)u_i(\lambda_1)
    \end{equation}
    and since $u_0(\cdot) = u_1(\cdot)$ and $u_N(\cdot) = u_{N+1}(\cdot)$, $w_{0}=w_{N}=0$. Thus, from \cref{eqn:forward_sum,eqn:Casoratian_diff} we have
    \begin{align}
      w_n &= w_0 + \sum_{i=0}^{n-1}\Delta w_i \\
	  &= -\Gamma \sum_{i=0}^{n-1}u_{i+1}(\lambda_2)u_{i+1}(\lambda_1) \label{eqn:cas_up} \\
	  &= \Gamma \sum_{i=n}^{N-1}u_{i+1}(\lambda_2)u_{i+1}(\lambda_1) \label{eqn:cas_down}.
    \end{align}
    Here, because $\mathbf{H}_W(\mathbb{P}_N)$ is a Jacobi matrix, we are free to choose $\mathbf{u}(\lambda_1)$ as everywhere positive and $\mathbf{u}(\lambda_2)$ as initially positive with no loss of generality. Further, it is known that $\mathbf{u}(\lambda_1)$ has no generalized zeros and $\mathbf{u}(\lambda_2)$ has exactly one, which we identify with $u_{\sigma}(\lambda_2)$. (See \textit{e.g.} Gantmakher.\cite{gantmakher2002oscillation}) Then, from \cref{eqn:cas_up}
    \begin{eqnarray}
      w_{n\leq \sigma} &=& -\Gamma \sum_{i=0}^{n-1}u_{i+1}(\lambda_2)u_{i+1}(\lambda_1) \\
		  &\leq& 0
    \end{eqnarray}

    Similarly, from \cref{eqn:cas_down}
    \begin{eqnarray}
      w_{n > \sigma} &=& \Gamma \sum_{i=n}^{N-1}u_{i+1}(\lambda_2)u_{i+1}(\lambda_1) \\
		&\leq& 0
    \end{eqnarray}
    so that we have $w_n \leq 0 \; \forall \; n \in \intset{0}{N}$.

    Finally, by \cref{eq:cas_temp} we arrive at
    \begin{equation}
      \Delta\left(\frac{u_i(\lambda_2)}{u_i(\lambda_1)} \right)\leq 0 \; \forall \; i \in \intset{0}{N}.
    \end{equation}
    Now, this sequence can be divided into three regions, where we will find that at least two of these regions are nonempty. Specifically, that this quantity is always decreasing guarantees that there exists some choice of $m < n \in \intset{1}{N}$ such that
    \begin{equation}\label{eqn:regions}
      \begin{cases}
    	\left(\frac{u_i(\lambda_2)}{u_i(\lambda_1)} \right) > 1, & i \in \intset{1}{m} \\
    	-1 \leq \left(\frac{u_i(\lambda_2)}{u_i(\lambda_1)} \right) \leq 1, & i \in \intset{m+1}{n} \\
    	\left(\frac{u_i(\lambda_2)}{u_i(\lambda_1)} \right) < -1, & i \in \intset{n+1}{N}
      \end{cases}
    \end{equation}
    and hence $\big(u_{i}^2(\lambda_2) - u_{i}^2(\lambda_1)\big)_{i=1}^N$ has at most two generalized zeros. Further, that $u_{i}(\lambda_2),u_{i}(\lambda_1)$ are normalized and orthogonal eigenvectors guarantees that $\big(u_{i}^2(\lambda_2) - u_{i}^2(\lambda_1)\big)_{i=1}^N$ has at least one generalized zero. Thus, our proof is complete.
  \end{proof}

  Using \Cref{lem:intersection} we now prove a discrete analogue of Lemma 3.2 from Lavine\cite{Lavine1994}:
  \begin{lemma}\label{lem:lavine}
	Let $\mathcal{W}$ be the set of convex potentials and $\mathcal{L}\subseteq{\mathcal{W}}$ be the set of linear potentials. Let $\mathbf{u}(\lambda_1)$, $\mathbf{u}(\lambda_2)$ be the two lowest eigenvectors of some $\mathbf{H}_{W}(\mathbb{P}_N)$ satisfying \cref{eqn:regions}.  Then, $\forall \;W\in\mathcal{W} \; \exists L \in \mathcal{L} \;\rvert\; \Gamma(\mathbf{H}_{W}(\mathbb{P}_N)) \geq \Gamma(\mathbf{H}_{L}(\mathbb{P}_N))$.
  \end{lemma}

  \begin{proof}
  	 Identify with $W(V_i)$ a convex function $w: \mathbb{R} \rightarrow \mathbb{R}$ such that $w(i)=W(V_i) \; \forall \; i \in \intset{1}{N}$. Then, we define the linear function $l_w: \mathbb{R} \rightarrow \mathbb{R}$ as
  	\begin{eqnarray}
  		l_w(i) = \frac{1}{n-m}\bigg((n-i)w(m) + (i-m)w(n)\bigg)
  	\end{eqnarray}
  	with $n$ and $m$ defined as in \Cref{lem:intersection}, and identify it with the corresponding $L_W \in \mathcal{L}$. Notably, $L_W(V_i) \leq W(V_i) \; \forall \; i \in \intset{1}{m}\cup\intset{n+1}{N}$ and $L(V_i) \geq W(V_i) \; \forall \; i \in \intset{m+1}{n}$. Then, clearly
	\begin{equation}
	  \langle\mathbf{W}-\mathbf{L}_W\rangle_{\mathbf{u}(\lambda_2)} - \langle\mathbf{W}-\mathbf{L}_W\rangle_{\mathbf{u}(\lambda_1)} \geq 0 \label{eq:39}
	\end{equation}
	where equality is obtained only when $W=L_W$.

	Now we consider the Schr\"{o}dinger operator that satisfies
	\begin{equation}
		\mathbf{H}_{W}(\mathbb{P}_N;\alpha) = \mathbf{H}_{W(\alpha)}(\mathbb{P}_N)
	\end{equation}
	and identify with $W(\alpha)$ the convex function $w(i;\alpha)$
	\begin{eqnarray}
		w(i;0) &=& w(i) \\
		\frac{d{w}}{d \alpha}(i;\alpha) &=& l_{w(\cdot;\alpha)}(i) - w(i;\alpha).
	\end{eqnarray}
	Thus, by \cref{eqn:Hellman-Feynman,eq:39} we have that the gap of $\mathbf{H}_{W(\alpha)}(\mathbb{P}_N)$ decreases with $\alpha$ and additionally that
	\begin{equation}\label{eqn:converge}
		w(i;\alpha) = e^{-\alpha}w(i) + \int_0^\alpha \frac{e^{s-\alpha}}{n(s)-m(s)}\bigg((n(s)-i)w\big(m(s);s\big) + (i-m(s))w\big(n(s);s\big)\bigg)ds.
	\end{equation}
	Hence, as $\alpha$ increases, we have that $w(i;\alpha)$ gets arbitrarily close to a linear function and therefore $W(\alpha)$ gets arbitrarily close to some function in $\mathcal{L}$.
  \end{proof}

  \subsection{The gap for linear potentials is lower bounded by the gap for constant potentials.}
    We start with $\mathbf{u}(\lambda_2),\mathbf{u}(\lambda_1)$ as the eigenvectors of $\mathbf{H}_{W}(\mathbb{P}_N)$ for some $W \in \mathcal{W}$. By \Cref{lem:lavine} we need only demonstrate that gaps associated with the class of linear potentials are lower bounded by the gaps associated with the constant potential. Because we are confined to a discrete setting, this takes a bit of work. The overall strategy is as follows: First, we restrict ourselves to a particular class of linear potentials and demonstrate that $\mathbf{u}(\lambda_1)$ is strictly decreasing. Then, we prove some facts about the ordering of the components of $\mathbf{u}(\lambda_2)$ around its node. Next, we demonstrate that for positive slopes, $\mathbf{u}(\lambda_2)$ always has a node left of center. These facts combine to complete our proof.

    We introduce the notation $\left[\mathbf{U}\right]_{ij} = (i-1)\delta_{ij}$ for the unit linear potential. Note that for any linear potential $L \in \mathcal{L}$ with slope $\alpha$, the potential $\alpha U$ has the same gap. Thus, we restrict our study to the unit potential multiplied by some parameter $\alpha$. Further, symmetry allows us to restrict ourselves to the case that $\alpha\geq0$.

    Our goal is to demonstrate that
    \begin{equation}\label{eqn:goal}
    	\frac{d \Gamma(\alpha)}{d \alpha} > 0
    \end{equation}
    for all $\alpha \geq 0$.

    We make use of the following lemma to reduce to the case that $u_1^2(\lambda_2) > u_1^2(\lambda_1)$:
    \begin{lemma}\label{lem:simplify}
    	Let $\alpha U\in\mathcal{L}$ where $U$ is the unit-linear potential. Then, for $\mathbf{H}_{\alpha U}(\mathbb{P}_N)$, if $u_1^2(\lambda_2) \leq u_1^2(\lambda_1)$, \cref{eqn:goal} is satisfied.
    \end{lemma}

    \begin{proof}
      By \cref{eqn:Hellman-Feynman},
      \begin{align}
      	\frac{d\Gamma(\alpha)}{d \alpha} &= \sum_{i=1}^{N}\left(u_i^2(\lambda_2)-u_i^2(\lambda_1)\right) (i-1)  \\
      	&= \sum_{i=1}^{N}\left(u_i^2(\lambda_2)-u_i^2(\lambda_1)\right) (i - c) \label{eqn:bla}
      \end{align}
      for any constant $c$. (Recall that the $\mathbf{u}(\lambda)$ are normalized eigenvectors.) From \Cref{lem:intersection} we know that if $u_1^2(\lambda_2) \leq u_1^2(\lambda_1)$ then $\exists n < N$ such that $u_i^2(\lambda_2)-u_i^2(\lambda_1) > 0$ for all $i>n$. Choosing $c = n$ we get that eq. (\ref{eqn:bla}) is non-negative for each term of the sum, thus completing the proof.
    \end{proof}

    Having reduced to the case that $u_1^2(\lambda_2) \geq u_1^2(\lambda_1)$, we now prove that $\mathbf{u}(\lambda_1)$ is a decreasing sequence:

    \begin{lemma}\label{lem:decreasing}
    	Let $\mathbf{H}_{\alpha U}(\mathbb{P}_N)$ be defined as in \Cref{lem:simplify}. Then, $\mathbf{u}(\lambda_1)$ is a decreasing sequence. Further, for $\alpha > 0$, $\mathbf{u}(\lambda_1)$ is strictly decreasing.
    \end{lemma}

    \begin{proof}
    First we note that at the boundaries, $\Delta u_0(\lambda_1) = \Delta u_N(\lambda_1) = 0$. Thus we know that the boundaries are local extrema of the $\mathbf{u}(\lambda_1)$-line. Now, we note that by \cref{eqn:recurrence}
    \begin{equation}
      \frac{u_2(\lambda_1)}{u_1(\lambda_1)} = (1 - \lambda_1) \leq 1
    \end{equation}
    where the inequality is strict for $\alpha > 0$ since this requires that $\lambda_1 > 0$. Thus, the $\mathbf{u}(\lambda_1)$-line is initially decreasing. Note that from \cref{eqn:recurrence_difference} when $\mathbf{W} = \mathbf{U}$, $\Delta^2 u_i(\lambda_1)$ has at most one sign change. Thus, the second boundary term cannot be a maximum and, therefore, both boundaries must be global extrema. We therefore have that $\mathbf{u}(\lambda_1)$ is decreasing for $\alpha \geq 0$ and strictly decreasing for $\alpha > 0$.
    \end{proof}

    We now recall a theorem by Cauchy and use it to derive an upper bound for $\lambda_2$:
    \begin{theorem}[Cauchy Interlace Theorem]\label{thm:Cauchy}
    	Let $\mathbf{A}$ be an $N\times N$ Hermitian matrix with eigenvalues $\lambda_1 \leq \lambda_2 \leq \dots \leq \lambda_N$. Suppose that B is an $(N-1)\times (N-1)$ principal submatrix of $\mathbf{A}$ with eigenvalues $\mu_1 \leq \mu_2 \leq \dots \leq \mu_{N-1}$. Then, the eigenvalues are ordered such that $\lambda_1 \leq \mu_1 \leq \lambda_2 \leq \mu_2 \leq \dots \leq \lambda_{N-1} \leq \mu_{N-1} \leq \lambda_N$.
    \end{theorem}
    \begin{proof}
      For proof, we refer the reader to Hwang. \cite{Hwang2004}
    \end{proof}
    \begin{lemma}\label{lem:upper}
    	Suppose that an $N \times N$ Hermitian matrix $\mathbf{A}$ with $N \geq 3$ has the $3\times 3$ principal submatrix
    	\[\mathbf{B}(\delta) = \left( \begin{array}{ccc}
	2 - \delta & -1 & 0 \\
	-1 & 2 + \alpha & -1 \\
	0 & -1 & 2 + 2\alpha
    	\end{array} \right) \]
    	with $\delta \geq 0$. Then, if $\lambda_2$ is the second lowest eigenvalue of $\mathbf{A}$, $\lambda_2 \leq 2 + \alpha$.
    \end{lemma}
    \begin{proof}
    	Let $\mu_1(\delta) \leq \mu_2(\delta) \leq \mu_3(\delta)$ be the eigenvalues of $\mathbf{B}(\delta)$. That $\lambda_2 \leq \mu_2(\delta)$ is obvious from repeated applications of \Cref{thm:Cauchy}. From, \Cref{thm:Hellman-Feynman},
    	\begin{equation}
    		\frac{d \mu_2(\delta)}{d \delta} \leq 0
    	\end{equation}
    	and by direct calculation, $\mu_2(0) = 2 + \alpha$. Thus, $\lambda_2 \leq 2 + \alpha$.
    \end{proof}

    \Cref{lem:upper} now combines with the following fact to give an ordering of the components of $\mathbf{u}(\lambda_2)$:
    \begin{lemma}\label{lem:convex_combinations}
    	Let $\mathbf{H}_{\alpha U}(\mathbb{P}_N)$ be defined as in \cref{lem:simplify} and let $\mathbf{u}(\lambda)$ be an eigenvector. Define the quantity
    	\begin{equation}\label{def:convex_combinations}
    		u_{i+\epsilon}(\lambda) \stackrel{\text{def}}{=} \epsilon u_{i+1}(\lambda) + (1-\epsilon)u_{i}(\lambda).
    	\end{equation}
	Then, for $u_{i}(\lambda)$ not a generalized zero,
	\begin{equation}\label{eqn:recurrence_cc}
	  u_{i+1+\epsilon}(\lambda) = \left(2+\alpha(j_{i+\epsilon}-1)-\lambda\right)u_{i+\epsilon} - u_{i-1+\epsilon}
	\end{equation}
	for some $j_{i+\epsilon} \in \left[i,i+1\right]$.
    \end{lemma}
    \begin{proof}
    	First, note that from \cref{eqn:recurrence}
    	\begin{equation}\label{eqn:cc_1}
    		u_{i+1+\epsilon}(\lambda) = \left(2-\lambda\right)u_{i+\epsilon}(\lambda) + \alpha \left(i \epsilon u_{i+1}(\lambda) + (i-1)(1-\epsilon)u_{i}(\lambda) \right) -  u_{i-1+\epsilon}(\lambda).
    	\end{equation}
	Now, with $\text{sign}(u_{i+1}(\lambda))=\text{sign}(u_i(\lambda))$, there exists a $j_{i+\epsilon} \in \left[i,i+1\right]$ such that
	\begin{equation}
		i \epsilon u_{i+1}(\lambda) + (i-1)(1-\epsilon)u_{i}(\lambda) = (j_{i+\epsilon}-1) \left(\epsilon u_{i+1}(\lambda) + (1-\epsilon)u_{i}(\lambda)\right).
	\end{equation}
	Thus, \cref{eqn:cc_1} becomes
	\begin{equation}
    		u_{i+1+\epsilon}(\lambda) = \left(2+\alpha (j_{i+\epsilon}-1) -\lambda\right)u_{i+\epsilon}(\lambda) -  u_{i-1+\epsilon}(\lambda).
    	\end{equation}
    \end{proof}
    \begin{lemma}[Ordering $\mathbf{u}(\lambda_2)$]\label{lem:ordering_2}
    	Let $\mathbf{H}_{\alpha U}(\mathbb{P}_N)$ be defined as in \Cref{lem:simplify}. Let $x$ represent the first node of the $\mathbf{u}(\lambda_2)$-line and let $u_m(\lambda_2)$ be the corresponding generalized zero. Suppose that $x \leq (N+1)/2$. Let $u_{i+\epsilon}(\lambda)$ be defined as in \Cref{lem:convex_combinations}. Then,
    	\begin{equation}
	-1 \geq \begin{cases}
 			\frac{u_{m+k+\epsilon}(\lambda_2)}{u_{m-1-k+\epsilon}(\lambda_2)} \; \; \text{for $m \leq x \leq m + \frac{1}{2}$ and $k \in \intset{0}{m-1}$} \\
 			\frac{u_{m+k+\epsilon}(\lambda_2)}{u_{m-k+\epsilon}(\lambda_2)} \; \; \text{for $m + \frac{1}{2} < x \leq m + 1$ and $k \in \intset{1}{m}$}.
 		\end{cases}
 	\end{equation}
    \end{lemma}

    \begin{proof}
      We proceed to prove this lemma by induction. First, consider the case that $m+1/2 \leq x < m+1$ for some $m\in \intset{1}{\lfloor N/2 \rfloor}$. For simplicity, let $\mathbf{u}(\lambda_2) = \mathbf{u}$. Then, there exists an $\epsilon$ such that $(1-\epsilon)u_{m}+\epsilon u_{m+1}=0$. So, from \cref{eqn:cc_1} we can consider the base case
      \begin{align}\label{eqn:base_1}
	u_{m+1+\epsilon} &= \epsilon \alpha u_{m+1} - u_{m-1+\epsilon} \\
	&< -u_{m-1+\epsilon}.
      \end{align}
      For the induction, rearrange \cref{eqn:recurrence_cc} for terms left and right of the node,
      \begin{equation}\label{eqn:temp6}
      	\frac{u_{m+k+2+\epsilon}+u_{m+k+\epsilon}}{u_{m+k+1+\epsilon}} - \frac{u_{m-k-2+\epsilon}+u_{m-k+\epsilon}}{u_{m-k-1+\epsilon}} = \alpha(j_{m+k+1+\epsilon}-j_{m-k-1+\epsilon}) > 0.
      \end{equation}
      Now assume
      \begin{equation}\label{eqn:assumption}
      	\frac{u_{m+k+\epsilon}}{u_{m+k+1+\epsilon}} \leq \frac{u_{m-k+\epsilon}}{u_{m-k-1+\epsilon}}
      \end{equation}
      thus, by \cref{eqn:temp6}
      \begin{equation}
      	\frac{u_{m+k+2+\epsilon}}{u_{m+k+1+\epsilon}} \geq \frac{u_{m-k-2+\epsilon}}{u_{m-k-1+\epsilon}}.
      \end{equation}
      Thus,
      \begin{equation}
      	\frac{u_{m+k+2+\epsilon}}{u_{m-k-2+\epsilon}} \leq \frac{u_{m+k+1+\epsilon}}{u_{m-k-1+\epsilon}}.
      \end{equation}
      Finally, taking $k=0$, \cref{eqn:base_1} satisfies \cref{eqn:assumption} and
      \begin{equation}
      	\frac{u_{m+k'+\epsilon}}{u_{m-k'+\epsilon}} \leq -1
      \end{equation}
      for all $k' \in \intset{1}{m}$.

      Next we consider the case that $m \leq x < m+1/2$. In this case, by \Cref{def:node} we can choose $\epsilon$ such that $u_{m+\epsilon} = -u_{m-1+\epsilon}$. Then,
      \begin{equation}\label{eqn:temp5}
      	\frac{u_{m+k+1+\epsilon}+u_{m+k-1+\epsilon}}{u_{m+k+\epsilon}} - \frac{u_{m-k-2+\epsilon}+u_{m-k+\epsilon}}{u_{m-k-1+\epsilon}} = \alpha(j_{m+k+\epsilon}-j_{m-k-1+\epsilon}) > 0.
      \end{equation}
      This time, assume
      \begin{equation}\label{eqn:assumption_2}
      	\frac{u_{m+k-1+\epsilon}}{u_{m+k+\epsilon}} \leq \frac{u_{m-k+\epsilon}}{u_{m-k-1+\epsilon}}
      \end{equation}
      then, by \cref{eqn:temp5}
      \begin{equation}
      	\frac{u_{m+k+1+\epsilon}}{u_{m+k+\epsilon}} \geq \frac{u_{m-k-2+\epsilon}}{u_{m-k-1+\epsilon}}.
      \end{equation}
      Hence,
      \begin{equation}
      	\frac{u_{m+k+1+\epsilon}}{u_{m-k-2+\epsilon}} \leq \frac{u_{m+k+\epsilon}}{u_{m-k-1+\epsilon}}.
      \end{equation}
      Again, taking $k=0$, we have that
      \begin{equation}
      	\frac{u_{m+k'+\epsilon}}{u_{m-1-k'+\epsilon}} \leq \frac{u_{m+\epsilon}}{u_{m-1+\epsilon}} \leq -1.
      \end{equation}
      for all $k' \in \intset{0}{m-1}$.
    \end{proof}

    Now, we recall a theorem due to Gantmakher and Kre\u{i}n:\cite{gantmakher2002oscillation}
    \begin{restatable}{theorem}{Gantmakher}\label{thm:Gantmakher}
      Let $\mathbf{u}(\mu;\alpha)$,$\mathbf{u}(\lambda;\beta)$ be two vectors of length $N$ satisfying \cref{eqn:recurrence} and with
      \begin{equation} \label{eqn:rest}
	      \Theta_{W,i}(\mu-\lambda;\alpha,\beta) \leq 0 \;\; \text{$\forall \; i \in \intset{m}{n}$}
      \end{equation}
      where $\Theta_{W,i}(\mu-\lambda;\alpha,\beta) < 0$ for at least some $i \in \intset{m}{n}$. We extend both vectors to length $N+2$ by including nodes at $u_0$ and $u_{N+1}$. (So long as \cref{eqn:recurrence} is satisfied, despite previous choices of $u_0$ and $u_{N+1}$, these points are always considered nodes.) Let $\eta \in [m-1,m), \xi \in (n,n+1]$ be two adjacent nodes of $\mathbf{u}(\lambda;\beta)$ with $m \leq n \in \intset{0}{N+1}$. Then there exists at least one node of $\mathbf{u}(\mu;\alpha)$ between $\eta$ and $\xi$.
    \end{restatable}
    \begin{proof}
    	This fact is adapted directly from Gantmakher and Kre\u{i}n, with modifications made to allow for our parameterization. The argument is provided in detail in \Cref{app:Gantmakher} for the unfamiliar reader.
    \end{proof}
  \begin{restatable}{lemma}{Node}\label{lem:node_left}
      Let $\mathbf{H}_{\alpha U}(\mathbb{P}_N)$ be defined as in \Cref{lem:simplify}. $\mathbf{u}(\lambda_2)$ always has a node at or left of $x=(N+1)/2$.
  \end{restatable}
    \begin{proof}
    We only want to consider variations with respect to one parameter, so we fix $\lambda=\mu_0,\alpha=\beta_0$. Then, we note that, by \cref{eqn:theta}, $\Theta_{U,i}$ is an increasing sequence in $i$. Now, we assume that there exists a node of $\mathbf{u}(\mu_0)$ at $x=(N+1)/2$. Next, at $\beta=\beta_0$ and $\mu=\mu_0$, $\Theta_{U,i}$ is identically $0$. Then,
      \begin{equation}\label{eq:theta_vary}
	\frac{d \Theta_{U,i}}{d\beta}\bigg\rvert_{\mu=\mu_0} = U_i - \frac{d \mu}{d \beta}\bigg\rvert_{\mu=\mu_0}= U_i - \langle \mathbf{U} \rangle_{\mathbf{u}(\mu_0)}
      \end{equation}
      Our assumption that $\mathbf{u}(\mu_0)$ at $x=(N+1)/2$ requires that $\epsilon=1$ in \Cref{lem:ordering_2}. Then, \Cref{lem:ordering_2} becomes an exact statement about the ordering of the components of $\mathbf{u}(\mu_0)$. Hence, $\langle \mathbf{U}\rangle_{\mathbf{u}(\mu_0)} \geq (N-1)/2$ and we have that
      \begin{align}\label{eq:node_theta}
	U_{\left\lfloor\frac{N+1}{2}\right\rfloor} - \left\langle \mathbf{U} \right\rangle_{\mathbf{u}(\mu_0)} &\leq U_{\left\lfloor\frac{N+1}{2}\right\rfloor} - \frac{N-1}{2} \\
	&= \left(\left\lfloor\frac{N+1}{2}\right\rfloor - 1\right) - \frac{N-1}{2} \\
	&\leq 0.
      \end{align}
      Further, in the same fashion
      \begin{equation}
      	U_{\left\lfloor\frac{N+1}{2}-1\right\rfloor} - \left\langle \mathbf{U} \right\rangle_{\mathbf{u}(\mu_0)} < 0.
      \end{equation}
      Hence, for some $\beta-\beta_0=\xi>0$ with $\xi$ sufficiently close to 0, $\Theta_{U,i}\leq0 \; \forall \; i \in \intset{1}{\lfloor(N+1)/2\rfloor}$ with at least some $i$ such that $\Theta_{U,i}<0$. Thus, at $\beta=\beta_0$ the node of $\mathbf{u}(\lambda_2)$ shifts left as $\beta$ increases. Note that if $\beta=\beta_0=0$, symmetry forces the node of the $\mathbf{u}(\mu)$-line to occur at $x=(N+1)/2$. Thus, there is initially a node at $(N+1)/2$ and whenever there is a node at $x=(N+1)/2$ it shifts left. Hence, there is always a node at or left of $(N+1)/2$.
    \end{proof}

    \begin{remark}
    	In fact, with some of the facts that follow, we demonstrate that the node shifts left with increasing $\alpha$. For proof, see \Cref{app:Sturm-Picone}.
    \end{remark}

    \Cref{lem:node_left} allows us to strengthen \Cref{lem:ordering_2} through the following fact:
        \begin{lemma}\label{lem:decreasing_2}
    	Let $\mathbf{H}_{\alpha U}(\mathbb{P}_N)$ be defined as in \Cref{lem:simplify}. Let $x$ represent the first node of the $\mathbf{u}(\lambda_2)$-line. Then, there exists a symmetric region $S=\intset{1}{m}$ about $x$ such that $\mathbf{u}(\lambda_2)$ is a decreasing sequence.
    \end{lemma}

    \begin{proof}
    	We begin by considering the first point after the node such that $\mathbf{u}(\lambda_2)$ is increasing and label it by $m$ such that $\Delta u_m (\lambda_2) \Delta u_{m-1}(\lambda_2) < 0$. Then,
    	\begin{equation}\label{eqn:decreasing_1}
    		u_{m+1}(\lambda_2) = \left(2 + \alpha (m-1) -\lambda_2\right) u_m(\lambda_2) - u_{m-1}(\lambda_2).
    	\end{equation}
    	Now, rearranging \cref{eqn:decreasing_1}
    	\begin{equation}\label{eqn:sub_path}
    		u_m(\lambda_2) = \left(2 + \left(1 - \frac{u_{m+1}(\lambda_2)}{u_{m}(\lambda_2)}\right) +\alpha (m-1) -\lambda_2\right) u_m(\lambda_2)-u_{m-1}(\lambda_2).
    	\end{equation}
    	Note that in \cref{eqn:sub_path}, because $u_m(\lambda_2)\leq u_{m+1}(\lambda_2)<0$, we know that $1 > u_{m+1}(\lambda_2)/u_m(\lambda_2)$ and thus $(u_i(\lambda_2))_{i\in\intset{1}{m}}$ is an eigenvector of  $\mathbf{H}_W(\graph{P}{m})$ where
    	\begin{equation}
	  W_i = \alpha U_i + \delta_{im}\left(1 - \frac{u_{m+1}(\lambda_2)}{u_{m}(\lambda_2)}\right).
	\end{equation}
	By \Cref{lem:node_left} the second eigenvector of $\mathbf{H}_{\alpha U}(\graph{P}{m})$ has a node left of center. Since $\lambda_2$ is greater than the second eigenvalue of $\mathbf{H}_{\alpha U}(\graph{P}{m})$ and we know that $\mathbf{W}$ is identical to $\mathbf{U}$ in all but the $m^{\text{th}}$ component, we have by \Cref{thm:Gantmakher} that $\left(u_i(\lambda_2)\right)_{i\in\intset{1}{m}}$ has a node left of center. Further, by our assumptions, $(u_i(\lambda_2))_{i\in\intset{1}{m}}$ is a decreasing sequence. Therefore, there exists a symmetric region $S=\intset{1}{m}$ about $x$ such that $\mathbf{u}(\lambda_2)$ is strictly decreasing.
    \end{proof}

    Using \Cref{lem:decreasing_2} we now prove a corollary to \Cref{lem:ordering_2} that holds regardless of whether the node falls directly on a vertex:
       \begin{corollary}\label{cor:ordering}
    	Let $\mathbf{H}_{\alpha U}(\mathbb{P}_N)$ and $x$ be defined as in \cref{lem:ordering_2}. Let $\mathbf{u}(\lambda_2)$ be a decreasing sequence. Then,
    	\begin{equation}
	-1 \geq \begin{cases}
 			\frac{u_{m+k+1}(\lambda_2)}{u_{m-k}(\lambda_2)} \; \; \text{for $m \leq x \leq m + \frac{1}{2} \forall k \in \intset{0}{m-1}$} \\
 			\frac{u_{m+2+k}(\lambda_2)}{u_{m-k}(\lambda_2)} \; \; \text{for $m + \frac{1}{2} < x \leq m + 1 \forall k \in \intset{1}{m}$}.
 		\end{cases}
 	\end{equation}
    \end{corollary}
    \begin{proof}
    	For the first case, assume that $m+1/2 \leq x < m+1$ for some $m\in \intset{1}{\lfloor N/2 \rfloor}$. Then, from \Cref{lem:ordering_2} we have that
    	\begin{equation}
    		-1 \geq \frac{u_{m+k+\epsilon}(\lambda_2)}{u_{m-k+\epsilon}(\lambda_2)} \; \; \text{for $m + \frac{1}{2} < x \leq m + 1 \forall k \in \intset{1}{m}$}.
    	\end{equation}
	Then, since $\mathbf{u}(\lambda_2)$ is decreasing, $u_{m+1+k} \leq u_{m+k+\epsilon}$ and also $u_{m-k+\epsilon} \geq u_{m+1-k}$. Thus,
	\begin{equation}
		-1 \geq \frac{u_{m+k+\epsilon}(\lambda_2)}{u_{m-k+\epsilon}(\lambda_2)} \geq \frac{u_{m+1+k}(\lambda_2)}{u_{m+1-k}(\lambda_2)}.
	\end{equation}

	Similarly, for the case that $m \leq x < m+1/2$, we have that
	\begin{equation}
    		-1 \geq \frac{u_{m+k+\epsilon}(\lambda_2)}{u_{m-1-k+\epsilon}(\lambda_2)} \; \; \text{for $m \leq x \leq m + \frac{1}{2} \forall k \in \intset{1}{m-1}$}.
    	\end{equation}
	In this case, we have that $u_{m-1-k+\epsilon} \geq u_{m-k}$. So that finally,
	\begin{equation}
    		-1 \geq \frac{u_{m+k+\epsilon}(\lambda_2)}{u_{m-1-k+\epsilon}(\lambda_2)} \geq \frac{u_{m+1+k}(\lambda_2)}{u_{m-k}(\lambda_2)}.
    	\end{equation}
    \end{proof}

    \begin{theorem}
      For $\mathbb{P}_N$,
      \begin{equation}
      	\Gamma_{W \in \mathcal{W}} \geq 2\left(1-\cos\left(\frac{\pi}{N}\right)\right).
      \end{equation}
    \end{theorem}

    \begin{proof}
      From \Cref{thm:Hellman-Feynman} we know that so long as $\langle \mathbf{U} \rangle_{\mathbf{u}(\lambda_2)} - \langle \mathbf{U} \rangle_{\mathbf{u}(\lambda_1)} > 0$, the gap is increasing.

      Consider a set of indices $S_m$ symmetric about $m$, the index corresponding to the first generalized zero $u_m(\lambda_2)$ of $\mathbf{u}(\lambda_2)$. Now, define $\mathbf{v}(\lambda_i) = \big(u(\lambda_i)\big)_{i\in S_m}$. From \Cref{lem:decreasing,lem:ordering_2} we know that
      \begin{equation}
      	\big\langle \mathbf{U} \big\rangle_{\mathbf{v}(\lambda_2)} \geq \big\langle \mathbf{U} \big\rangle_{\mathbf{v}(\lambda_1)}.
      \end{equation}
      where we restrict $\mathbf{U}$ to the same number of terms as $\mathbf{v}(\lambda_2)$.

      By \Cref{lem:node_left} we know that the node of the $\mathbf{u}(\lambda_2)$-line must occur at or before the midpoint of $\mathbf{u}(\lambda_2)$. Thus, $S_m$ can be taken as $S_m = \intset{1}{2m+1}$. By \Cref{lem:simplify} we restrict ourselves to the case that $u_1^2(\lambda_2) > u_1^2(\lambda_1)$. With this restriction, \Cref{lem:ordering_2,lem:decreasing} insist that $u_k^2(\lambda_2) > u_k^2(\lambda_1) \; \forall k \in \intset{1}{N}/S_m$. It is then obvious that
      \begin{equation}
      	\big\langle \mathbf{U} \big\rangle_{\mathbf{u}(\lambda_2)} \geq \big\langle \mathbf{U} \big\rangle_{\mathbf{u}(\lambda_1)}
      \end{equation}
      for $\alpha \geq 0$. Thus, we know that $\Gamma$ is at a minimum for $\alpha=0$. Now, at $\alpha=0$ we find that $\lambda_1=0$ and thus $\Gamma = \lambda_2$. Hence,
      \begin{equation}
	\Gamma_{W \in \mathcal{W}} \geq 2\left(1-\cos\left(\frac{\pi}{N}\right)\right).
      \end{equation}
    \end{proof}

    \section{The Hypercube Graph}
    In this section we find a tight lower bound for the gap for Hamming-symmetric convex potentials on the N-dimensional hypercube graph $\graph{H}{2^N} = (V,E)$. To define $\graph{H}{2^N}$, we identify with each vertex $V_i \in V$ a unique vector $\mathbf{b}_i \in \{0,1\}^N$. Then, we choose $E = \{(V_i,V_j)\; \big\rvert \;|\mathbf{v}_i-\mathbf{v}_j| = 1 \}$, where $|\cdot|$ here denotes the 1-norm. (In the language of computer science, $\mathbf{b}_i, \mathbf{b}_j$ are bit-strings and $|\mathbf{b}_i - \mathbf{b}_j|$ is their Hamming distance.) As in (\ref{eqn:hw}) the Schr\"{o}dinger operator includes a potential term $\mathbf{W}$. Thus, an eigenvector $u(\lambda)$ of eigenvalue $\lambda$ satisfies
    \begin{equation}\label{eqn:recurrence_hyper_prior}
    	(N + W_i - \lambda)u_i(\lambda) = \sum_{(V_i,V_j)\in E} u_j(\lambda) \;\; \text{for $V_i \in V$.}
    \end{equation}
    Here we restrict our attention to the case that the potential depends only on Hamming distance from the vertex of minimum potential. We can label this minimum by the all zeros string, and therefore $W_i = W_{|b_i|}$. In this case, the set of Hamming-symmetric vectors are an invariant subspace of the Schr\"{o}dinger operator.
    \begin{remark}
    In the language of quantum-mechanics, this is the space spanned by the $N+1$ state vectors that are uniform superpositions over bit-strings of a given Hamming weight. By Schr{\"o}dinger's equation, no time-evolution induced by a (possibly time-dependent) Hamming symmetric Hamiltonian will ever drive transitions out of this subspace. For many cases, it is only the gap within this subspace that is of interest.
    \end{remark}
    Below, we will bound the gap within the Hamming-symmetric subspace. Here, the (normalized) uniform superpositions over bit-strings of each Hamming weight form an orthonormal basis for this subspace.    Given a state-vector $\mathbf{u}(\lambda)$, let $v_m(\lambda)$ denote the inner product of $\mathbf{u}(\lambda)$ with the Hamming-weight-$m$ basis vector. That is,
    \begin{equation}
    	v_m(\lambda) = \frac{1}{\sqrt{\binom{N}{m}}} \sum_{|b_i|=m} u_i(\lambda).
    \end{equation}

    Because $\mathbf{u}(\lambda)$ lies within the symmetric subspace, this corresponds to rewriting the vector in a different basis. For arbitrary vectors in the full Hilbert space, this would be a projection onto the symmetric subspace.

Then, with a bit of work, \cref{eqn:recurrence_hyper_prior} becomes
    \begin{equation}\label{eqn:recurrence_hyper}
    	(N+W_m-\lambda)v_m(\lambda) = h(m-1)v_{m-1}(\lambda) + h(m)v_{m+1}(\lambda)
    \end{equation}
    where
    \begin{equation}
    	h(m) = \sqrt{(m+1)(N-m)}.
    \end{equation}

    Now, we know that \cref{eqn:recurrence_hyper} is the recurrence relation satisfied by some Jacobi matrix $\mathbf{J}$ with eigenvalue $\lambda \in (\lambda_i)_{i=1}^N$. In keeping with our typical ordering, we choose $\lambda_1 < \lambda_2 < \dots <\lambda_N$. Further, since we know that we can shift the diagonal by any $c \mathbbm{1}, c \in \mathbb{R}$ without altering the gap, we instead consider $\mathbf{J}\rightarrow \mathbf{J} -N\mathbbm{1}$ which satisfies the recurrence relation
    \begin{equation}\label{eqn:recurrence_hyper2}
    	(W_m-\lambda)v_m(\lambda) = h(m-1)v_{m-1}(\lambda) + h(m)v_{m+1}(\lambda)
    \end{equation}
    without any loss of generality.

    \begin{remark}
    	The reader should note that unlike \Cref{sec:path}, $v_0(\lambda)$ is not a boundary term, but $v_{N+1}(\lambda)$ is. This inconsistency is an artifact of labeling vertices by their Hamming weights as there are vertices with Hamming weight $0$, but none with Hamming weight $N+1$. The boundary terms will be defined where appropriate.
    \end{remark}

    Now, we define the transformation,
    \begin{equation}\label{eqn:trans}
    	v'_m(\lambda) \stackrel{\text{def}}{=} f(m) v_m(\lambda)
    \end{equation}
    where $f(m)$ is given by
    \begin{equation}
      f(m) \stackrel{\text{def}}{=} \begin{cases}
      f_0 \prod_{\stackrel{j\in\text{Odd}}{0<j<m}}\frac{h(j-1)}{h(j)} & \text{if $m$ is even} \\
      f_1 \prod_{\stackrel{j\in\text{Even}}{0\leq j<m}}\frac{h(j-1)}{h(j)} & \text{if $m$ is odd} \\
      \end{cases}
    \end{equation}
    and we choose,
    \begin{equation}
    	\frac{f_1}{f_0}= \frac{\sqrt{N}\;\Gamma(N)}{2^{N-1}(\Gamma(\frac{N+1}{2}))^2}
    \end{equation}
    where the $\Gamma$ above represents the gamma function, not the gap.

    With this transformation, we have from \cref{eqn:recurrence_hyper2,eqn:trans}
    \begin{equation}\label{eqn:recurrence_hyper3}
    	v'_{m-1}(\lambda)-2v'_m(\lambda)+v'_{m+1}(\lambda) = \frac{f(m)}{h(m)}{f(m+1)}(W_m - q_m -\lambda)v'_m(\lambda)
    \end{equation}
     where
    \begin{equation}
    	q_m \stackrel{\text{def}}{=} \frac{2 h(m)f(m+1)}{f(m)}.
    \end{equation}
    Here, our choice of $q_m$ (alternatively our choice of $f_1/f_0$) maintains symmetry and consistency across various choices of $N$.

    Now we consider the Casoratian sequence corresponding to \cref{eqn:Casoratian}.
    \begin{equation}
    	w_i(\mathbf{v'}(\lambda_2),\mathbf{v'}(\lambda_1)) = v'_{i+1}(\lambda_2)v'_i(\lambda_1) - v'_{i+1}(\lambda_1)v'_i(\lambda_2)
    \end{equation}

    For consistency, we choose $v'_{-1}(\cdot)=v'_{N+1}(\cdot)=0$ and we get that $w_{-1}=w_{N+1}=0$. Then, from \cref{eqn:recurrence_hyper3}, similarly to \cref{eqn:Casoratian_diff}, we have that
    \begin{equation}
    	\Delta w_{k-1} = -\Gamma\frac{f(k)}{h(k)f(k+1)}v'_k(\lambda_2)v'_k(\lambda_1).
    \end{equation}
    We note that since $\mathbf{v}(\cdot)$ are the eigenvectors of a Jacobi matrix, $\mathbf{v}(\lambda_1)$ has no generalized zeros and $\mathbf{v}(\lambda_2)$ has precisely one generalized zero. Then, since $f(m)>0 \; \forall m \; \in \; \intset{0}{N}$ we know that $\mathbf{v'}(\lambda_1)$ has no zeros and $\mathbf{v'}(\lambda_2)$ has precisely one. Thus, labeling the generalized zero of $\mathbf{v'}(\lambda_2)$ by $n$, we have that
    \begin{eqnarray}
    	w_{m\leq n} &=& \sum_{k=-1}^{m-1}\Delta w_i						\\
    	&=&-\Gamma \sum_{k=-1}^{k-1}\frac{f(k)}{h(k)f(k+1)}v'_k(\lambda_2)v'_k(\lambda_1)	\\
    	&<& 0
    \end{eqnarray}
    and similarly
    \begin{eqnarray}
    	w_{m > n} &=& -\sum_{k=m}^{N}\Delta w_i						\\
    	&=& \Gamma \sum_{k=m}^{N}\frac{f(k)}{h(k)f(k+1)}v'_k(\lambda_2)v'_k(\lambda_1)	\\
    	&<& 0
    \end{eqnarray}
    so that $w_i < 0 \; \forall \; i \in \intset{0}{N}$. As we have already seen in \cref{eqn:regions}, this guarantees that $\left({v'}_i^2(\lambda_2)-{v'}_i^2(\lambda_1)\right)_{i=0}^{N}$ has at most two generalized zeros. Now, because we have that
    \begin{equation}
    	{v'}_i^2(\lambda_2)-{v'}_i^2(\lambda_1) = f(k)^2\left(v_i^2(\lambda_2)-v_i^2(\lambda_1)\right)
    \end{equation}
    ${v'}_i^2(\lambda_2)-{v'}_i^2(\lambda_1)$ has the same sign as $v_i^2(\lambda_2)-v_i^2(\lambda_1)$ and thus, $\left(v_i^2(\lambda_2)-v_i^2(\lambda_1)\right)_{i=0}^{N}$ has at most two generalized zeros. That $\mathbf{v}(\lambda_2)$ is orthogonal to $\mathbf{v}(\lambda_1)$ guarantees that it has at least one generalized zero.

    At this point, we have satisfied the necessary conditions to apply an obvious analogue of \Cref{lem:lavine}:
    \begin{lemma}\label{lem:lavine2}
	Let $\mathcal{W}$ be the set of convex potentials and $\mathcal{L}\subseteq{\mathcal{W}}$ be the set of linear potentials. Let $\mathbf{u}(\lambda_1)$, $\mathbf{u}(\lambda_2)$ satisfying \cref{eqn:regions} be real-valued eigenvectors corresponding to the two lowest eigenvalues of some matrix $\mathbf{H}_{W}(\graph{P}{N}) + \mathbf{M}$ with real eigenvalues, where $\mathbf{M}$ is an arbitrary $N\times N$ matrix independent of $W$. Then, $\forall \;W\in\mathcal{W} \; \exists L \in \mathcal{L} \;\rvert\; \Gamma(\mathbf{H}_{W}+\mathbf{M}) \geq \Gamma(\mathbf{H}_{L}+\mathbf{M})$.
  \end{lemma}
  \begin{proof}
  	We note that because \Cref{lem:lavine} depends only upon the variational term $\mathbf{W}(\alpha)$, when some matrix some matrix $\mathbf{H}_{W}(\graph{P}{N}) + \mathbf{M}$ satisfies \cref{eqn:regions}, the proof is identical to that of \Cref{lem:lavine}. Therefore, this proof is omitted.
  \end{proof}
   The reduced Hamming-symmetric matrix corresponding to \cref{eqn:recurrence_hyper} is equivalent to $\mathbf{H}_{W}(\graph{P}{N+1}) + \mathbf{M}$ for some choice of $\mathbf{M}$. Thus, by \Cref{lem:lavine2} it has a lower bound for a linear, Hamming-symmetric potential.  Now, for such a linear potential we can consider $\alpha L_i = \alpha (i-N/2)$. Here, the eigenvalues are exactly solvable and given by
    \begin{equation}
    	\lambda_k = k\sqrt{4+\alpha^2} \;\;\; \text{$\forall k \in \{-N/2,-(N-1)/2,\dots,(N-1)/2,N/2\}$.}
    \end{equation}

    Then,
    \begin{equation}
    	\Gamma_{\alpha L \in \mathcal{L}} = \sqrt{4+\alpha^2}
    \end{equation}
    which is clearly minimized for $\alpha=0$. Thus, for convex, Hamming-symmetric potentials on the hypercube
    \begin{equation}
    	\Gamma \geq 2
    \end{equation}
    within the Hamming-symmetric subspace.

    \section{Acknowledgements}
    This work was supported in part by the joint center for Quantum Information and Computer Science (QuICS), a collaboration between the University of Maryland Institute for Advanced Computer Studies (UMIACS) and the NIST Information Technology Laboratory (ITL). Portions of this paper are a contribution of NIST, an agency of the US government and are not subject to US copyright.

\bibliography{Gap}

 \appendix

\section{A sufficient condition for some node of $\mathbf{u}(\mu)$ to separate adjacent nodes of $\mathbf{u}(\lambda)$}\label[secinapp]{app:Gantmakher}

In this section we give proof of \Cref{thm:Gantmakher}. The following proof is adapted directly from Gantmakher and Kre\u{i}n\cite{gantmakher2002oscillation}, however we consider vectors $\mathbf{u}(\lambda)$ and $\mathbf{u}(\mu)$ which need not be eigenvectors of the same matrix. Rather, we require only that these vectors satisfy \cref{eqn:recurrence}.

    \Gantmakher*

    \begin{proof}
      First, we consider the extension of vectors $\mathbf{u}(\mu;\alpha)$,$\mathbf{u}(\lambda;\beta)$. From \cref{eqn:recurrence}  we take $W_1 \mapsto W_1+1$ and, for $\mathbf{u}(\lambda;\beta)$ get
      \begin{equation}\label{eqn:recurrence_node}
      	\left(1+W_1 - \lambda\right)u_1(\lambda) = u_{2}(\lambda) + u_{0}(\lambda)
      \end{equation}
      and similarly for $\mathbf{u}(\mu;\alpha)$. Here, to maintain consistency between \cref{eqn:recurrence,eqn:recurrence_node} we require that in \cref{eqn:recurrence_node} $u_0=0$. Thus, $u_0$ is a node of $\mathbf{u}$. We similarly treat $u_{N+1}$ as a node. Further, since we have shifted $W_1,W_{N}$ by constants, \cref{eqn:theta} is unaltered. Hence, \cref{eqn:Casoratian} is unchanged and we can proceed with the proof.

      Let $\eta \in [m-1,m)$ and $\xi \in (n,n+1]$ be successive nodes of $\mathbf{u}(\lambda;\beta)$ with $\eta < \xi$. Without loss of generality, we assume that $u_i(\lambda;\beta) > 0 \; \forall \; i \in \intset{m}{n}$. Then,

      \begin{equation}\label{eqn:Appendix_conditional}
	\begin{cases}
	      (m-\eta)u_{m-1}(\lambda;\beta) + (\eta - m + 1)u_{m}(\lambda;\beta) &= 0 \\
	      (n+1-\xi)u_n(\lambda;\beta) + (\xi-n)u_{n+1}(\lambda;\beta) &= 0
	\end{cases}
      \end{equation}

      Now, again without loss of generality, we assume that $u_i(\mu;\alpha) > 0 \; \forall \; i \in \intset{m}{n}$. Hence, if $\mathbf{u}(\mu;\alpha)$ also has no nodes in $(m,n)$, we get that

      \begin{equation}\label{eqn:Appendix_conditional2}
	\begin{cases}
	      (m-\eta)u_{m-1}(\mu;\alpha) + (\eta-m+1)u_{m}(\mu;\alpha) &\geq 0 \\
	      (n+1-\xi)u_n(\mu;\alpha) + (\xi-n)u_{n+1}(\mu;\alpha) &\geq 0
	\end{cases}
      \end{equation}

      Combining \cref{eqn:Appendix_conditional,eqn:Appendix_conditional2} yields the inequalities
      \begin{align}
	      \label{eqn:Appendix_inequalities1} w_{m-1}\big(\mathbf{u}(\mu;\beta),\mathbf{u}(\lambda;\alpha)\big) &\leq 0 \\
	      \label{eqn:Appendix_inequalities2} w_{n}\big(\mathbf{u}(\mu;\beta),\mathbf{u}(\lambda;\alpha)\big) &\geq 0
      \end{align}

      Recall from \cref{eqn:Casoratian} that
      \begin{equation}
	      \Delta w_{i-1}\big(\mathbf{u}(\mu;\beta),\mathbf{u}(\lambda;\alpha)\big) = \Theta_{W,i}(\mu-\lambda;\beta,\alpha) u_{i}(\mu;\beta)u_i(\lambda;\alpha)
      \end{equation}
      where by summing both sides,
      \begin{equation}\label{eq:Appendix_Casoratian}
	      w_{n}\big(\mathbf{u}(\mu;\beta),\mathbf{u}(\lambda;\alpha)\big)-w_{m-1}\big(\mathbf{u}(\mu;\beta),\mathbf{u}(\lambda;\alpha)\big) = \sum_{i=m}^{n}\Theta_{W,i}(\mu-\lambda;\beta,\alpha) u_{i}(\mu;\beta)u_i(\lambda;\alpha)
      \end{equation}
      Thus, by \cref{eqn:Appendix_inequalities1,eqn:Appendix_inequalities2} we have that the left-hand side of \cref{eq:Appendix_Casoratian} is non-negative. Then, by our choice of $u_i(\lambda;\alpha),u_i(\mu;\beta)>0 \; \forall \; i \in \intset{m}{n}$, we see that if $\Theta_{W,i}(\mu-\lambda;\beta,\alpha) \leq 0 \; \forall \; i \in \intset{m}{n}$ with at least some $i\in \intset{m}{n}$ such that $\Theta_{W,i}(\mu-\lambda;\beta,\alpha) < 0$ we arrive at a contradiction.
    \end{proof}

\section{For $\mathbf{H}_{\alpha U}(\mathbb{P}_N)$, the node of $\mathbf{u}(\lambda_2)$ shifts left with increasing $\alpha$.}\label[secinapp]{app:Sturm-Picone}

\begin{theorem}
	Let $\mathbf{H}_{\alpha U}(\mathbb{P}_N)$ be defined as in \Cref{lem:node_left}. Then, the node of $\mathbf{u}(\lambda_2)$ shifts left with increasing $\alpha$.
\end{theorem}
\begin{proof}
	The proof proceeds in analogy to \Cref{lem:node_left}. First, note that by \Cref{cor:ordering},
	\begin{equation}
		\left\langle \mathbf{U} \right\rangle_{\mathbf{u}(\lambda_2)} \geq m-1
	\end{equation}
	where $m$ corresponds to the generalized zero of $\mathbf{u}(\lambda_2)$. Then, like \cref{eq:node_theta}
	\begin{align}
		U_{m} - \left\langle \mathbf{U} \right\rangle_{\mathbf{u}(\lambda_2)} &= (m-1) - \left\langle \mathbf{U} \right\rangle_{\mathbf{u}(\lambda_2)} \leq 0.
	\end{align}
	Note that because $\mathbf{u}(\lambda_2)$ has at least one positive and one negative term, the inequality is strict when $m=1$. If $m>1$,
	\begin{equation}
		U_{m-1} - \left\langle \mathbf{\lambda_2} \right\rangle_{\mathbf{u}(\lambda_2)} < 0.
	\end{equation}
	Thus, by \cref{eq:theta_vary}
	\begin{equation}
	  \frac{d \Theta_{U,i}}{d\beta} <0 \; \; \; \text{$\forall i \in \intset{1}{m}$}.
	\end{equation}
	Hence, by the same logic as \Cref{lem:node_left}, \Cref{thm:Gantmakher} applies and the node always shifts left with increasing $\alpha$.
\end{proof}


\end{document}